    \documentclass[12pt,a4paper]{article}
    \usepackage[cp1251]{inputenc}
    \usepackage[english]{babel}
    \usepackage{amsmath,amssymb,amsthm}
    \usepackage{algorithm}

\DeclareMathOperator{\Clo}{Clo}

\DeclareMathOperator{\proj}{pr}

\DeclareMathOperator{\Expanded}{ExpCov}

\DeclareMathOperator{\ConOne}{Con}

\DeclareMathOperator{\CSP}{CSP}

\newcommand{\cover}[1]{
{{#1}^{*}}
}

\renewcommand{\le}{\leqslant}
\renewcommand{\ge}{\geqslant}

\theoremstyle{definition}

\theoremstyle{plain}
\newtheorem{thm}{Theorem}[section]

\newtheorem{step}{Step}

\newtheorem{lem}[thm]{Lemma}

\begin{document}

\title{A modification of the CSP algorithm for infinite languages}

\author{Dmitriy Zhuk\\
Department of Mechanics and Mathematics\\
Moscow State University\\
Moscow, Russia
}
\date{}
\maketitle

\begin{abstract}
Constraint Satisfaction Problem on finite sets is known to be NP-complete in general but 
certain restrictions on the constraint language can ensure tractability.
It was proved \cite{BulatovProofCSP,MyProofCSP} that if a constraint language has a weak near unanimity polymorphism
then the corresponding constraint satisfaction problem is tractable, otherwise it is NP-complete. 
In the paper we present a modification of the algorithm from \cite{MyProofCSP} that works in polynomial time even for infinite constraint languages.
\end{abstract}

\section{Introduction}%
Formally, the \emph{Constraint Satisfaction Problem (CSP)} is defined as a triple $\langle \mathbf{X} , \mathbf{D} , \mathbf{C} \rangle$,
where
\begin{itemize}
\item
$\mathbf{X}=\{x_1,\ldots ,x_n\}$ is a set of variables,
\item
$\mathbf{D}=\{D_{1},\ldots ,D_{n}\}$ is a set of the respective domains,
\item
$\mathbf{C}=\{C_{1},\ldots ,C_{m}\}$ is a set of constraints,
\end{itemize}
where
each variable $x_{i}$ can take on values in the nonempty domain $D_{i}$,
every \emph{constraint} $C_{j}\in \mathbf{C}$ is a pair
$(t_{j},\rho_{j})$ where
$t_{j}$ is a tuple of variables of length $m_{j}$, called the \emph{constraint scope},
and $\rho_{j}$ is an $m_{j}$-ary relation on the corresponding domains,
called the \emph{constraint relation}.

The question is whether there exists \emph{a solution} to
$\langle \mathbf{X} , \mathbf{D} , \mathbf{C} \rangle$,
that is a mapping that assigns a value from $D_{i}$ to every variable $x_{i}$
such that
for each constraints $C_{j}$ the image of the constraint scope is a member of the constraint relation.

In this paper we consider only CSP over finite domains.
The general CSP is known to be NP-complete \cite{Num26, Num30}; however, certain restrictions
on the allowed form of constraints involved may ensure tractability (solvability in polynomial time)
\cite{Num4,Num20,Num22,Num23,CSPconjecture,BulatovAboutCSP}.
Below we provide a formalization to this idea.

To simplify the presentation we assume that
all the domains $D_{1},\ldots,D_{n}$ are subsets of a finite set $A$.
By $R_{A}$ we denote the set of all finitary relations on $A$,
that is, subsets of $A^{m}$ for some $m$.
Then all constraint relations can be viewed as relations from $R_{A}$.

For a set of relations $\Gamma\subseteq R_{A}$ by $\CSP(\Gamma)$
we denote the Constraint Satisfaction Problem where all
the constraint relations are from $\Gamma$.
The set $\Gamma$ is called \emph{a constraint language}.
Another way to formalize the Constraint Satisfaction Problem is
via conjunctive formulas.
Every $h$-ary relation on $A$ can be viewed as
a predicate, that is, a
mapping $A^{h}\rightarrow \{0,1\}$.
Suppose $\Gamma\subseteq R_{A}$, then $\CSP(\Gamma)$ is the following decision problem:
given a formula
$$\rho_{1}(x_{1,1},\ldots,x_{1,n_{1}})
\wedge
\dots
\wedge
\rho_{s}(x_{s,1},\ldots,x_{1,n_{s}})$$
where $\rho_{i}\in \Gamma$ for every $i$;
decide whether this formula is satisfiable.

It is well known that many combinatorial problems can be expressed as $\CSP(\Gamma)$
for some constraint language $\Gamma$.
Moreover, for some sets $\Gamma$ the corresponding decision problem can be solved in polynomial time;
while for others it is NP-complete.
It was conjectured that
$\CSP(\Gamma)$ is either in P, or NP-complete \cite{FederVardi}.

An operation $f$ is called \emph{idempotent} if $f(x,x,\ldots,x) = x$.
An operation $f$ is called \emph{a weak near-unanimity operation (WNU)} if
$f(y,x,\ldots,x) = f(x,y,x,\ldots,x) = \dots = f(x,x,\ldots,x,y).$

In the paper we present a modification of the algorithm from \cite{MyProofCSP}
that works in polynomial time for infinite constraint languages and therefore
prove CSP dichotomy conjecture for infinite constraint languages.

\begin{thm}\label{maintheorem}
Suppose $\Gamma\subseteq R_{A}$ is a set of relations.
Then $\CSP(\Gamma)$ can be solved in polynomial time if there exists a WNU
preserving $\Gamma$;
$CSP(\Gamma)$ is NP-complete otherwise.
\end{thm}

Note that the algorithm presented in \cite{BulatovProofCSP} also works for infinite constraint languages.

The paper is organized as follows.
In Section~\ref{Definition} we give all necessary definitions,
in Section~\ref{Algorithm}
we explain the algorithm starting with the new ideas.
In Section~\ref{CorretnessSection} we prove statements that show the correctness of the algorithm.

\section{Definitions}\label{Definition}
A set of operations is called \emph{a clone} if it is closed under composition and contains all projections.
For a set of operations $M$ by $\Clo(M)$ we denote the clone generated by $M$.

An idempotent WNU $w$ is called \emph{special} if $x \circ (x \circ y) = x \circ y$, where
$x \circ y = w(x,\dots,x,y)$.
It is not hard to show that for any idempotent WNU $w$ on a finite set there exists a special WNU $w'\in\Clo(w)$
(see Lemma 4.7 in \cite{miklos}).

A relation $\rho \subseteq A_{1}\times\dots\times A_{n}$ is called \emph{subdirect} if
for every $i$ the projection of $\rho$ onto the $i$-th coordinate is $A_{i}$.
For a relation $\rho$ by $\proj_{i_1,\ldots,i_{s}}(\rho)$
we denote the projection of $\rho$ onto the coordinates
$i_1,\ldots,i_{s}$.

\textbf{Algebras.}
\emph{An algebra} is a pair $\mathbf{A}:=(A;F)$, where $A$ is a finite set, called \emph{universe},
and $F$ is a family of operations on $A$, called \emph{basic operations of $\mathbf{A}$}.
In the paper we always assume that we have a special WNU preserving all constraint relations.
Therefore, every domain $D$ can be viewed as an algebra $(D;w)$.
By $\Clo(\mathbf{A})$ we denote the clone generated by all basic operations of $\mathbf{A}$.



\textbf{Congruences.}
An equivalence relation $\sigma$ on the universe of an algebra $\mathbf{A}$ is called \emph{a congruence}
if it is preserved by every operation of the algebra.
A congruence (an equivalence relation) is called \emph{proper}, if it is not equal 
to the full relation $A\times A$.
We use standard universal algebraic notions of term operation, subalgebra,  factor algebra, product of algebras,
see~\cite{bergman2011universal}.
We say that a subalgebra $\mathbf{R} = (R;F_R)$ is
\emph{a subdirect subalgebra} of $\mathbf{A}\times \mathbf{B}$
if $R$ is a subdirect relation in $A\times B$.

We say that the $i$-th variable of a relation $\rho$ is \emph{compatible with an equivalence relation $\sigma$}
if $(a_{1},\ldots,a_{n})\in\rho$ and $(a_{i},b_{i})\in\sigma$
implies
$(a_{1},\ldots,a_{i-1},b_{i},a_{i+1},\ldots,a_{n})\in\rho$.
We say that a relation is \emph{compatible} with $\sigma$ if every variable of this relation is compatible with $\sigma$.

For a relation $\rho$ by $\ConOne(\rho,i)$
we denote the binary relation $\sigma(y,y')$ defined by
$$\exists x_{1}\dots\exists x_{i-1}\exists x_{i+1}\dots\exists x_{n}\;\rho(x_{1},\ldots,x_{i-1},y,x_{i+1},\ldots,x_{n})\wedge
\rho(x_{1},\ldots,x_{i-1},y',x_{i+1},\ldots,x_{n}).$$
For a constraint $C = ((x_{1},\ldots,x_{n}),\rho)$,
by $\ConOne(C,x_{i})$ we denote $\ConOne(\rho,i)$.


\textbf{Essential and critical relations.}
A relation $\rho$ is called \emph{essential} if it cannot be represented as a conjunction of relations with smaller arities.
It is easy to see that
any relation $\rho$ can be represented as a conjunction of essential relations.
A relation $\rho\subseteq A_{1}\times\dots\times A_{n}$ is called \emph{critical}
if it cannot be represented as an intersection of other subalgebras
of $\mathbf A_{1}\times\dots\times \mathbf A_{n}$
and it has no dummy variables.

A tuple $(a_{1},a_{2},\ldots,a_{n})$ is called \textit{essential for an $n$-ary relation $\rho$}
if $(a_{1},a_{2},\ldots,a_{n})\notin\rho$ and
there exist $b_{1},b_{2},\ldots,b_{n}$ such that
$(a_{1},\ldots,a_{i-1},b_{i},a_{i+1},\ldots,a_{n})\in\rho$
for every $i\in \{1,2,\ldots,n\}$.

It is not hard to check the following lemma.
\begin{lem}\label{sushnabor}\cite{dm_post, MinimalClones, mvlsc}
Suppose $\rho\subseteq A^{n}$, where $n\ge 1$. Then the following conditions are equivalent:
\begin{enumerate}
\item $\rho$ is an essential relation;
\item there exists an essential tuple for $\rho$.
\end{enumerate}
\end{lem}

\textbf{Parallelogram property.}
We say that a relation $\rho$ \emph{has the parallelogram property}
if any permutation of its variables gives a relation $\rho'$ satisfying
$$\forall \alpha_{1},\beta_{1},\alpha_2,\beta_2\colon (\alpha_{1}\beta_2,\beta_1\alpha_2,\beta_1\beta_2\in\rho'
\Rightarrow \alpha_1\alpha_2\in\rho').$$

We say that \emph{the $i$-th variable of a relation $\rho$ is rectangular}
if for every $(a_{i},b_{i})\in\ConOne(\rho,i)$ and $(a_{1},\ldots,a_{n})\in\rho$
we have $(a_{1},\ldots,a_{i-1},b_{i},a_{i+1},\ldots,a_{n})\in\rho$.
We say that a relation is \emph{rectangular} if all of its variables are rectangular.
The following facts can be easily seen:
if the $i$-th variable of $\rho$ is rectangular then $\ConOne(\rho,i)$ is a congruence;
if a relation has the parallelogram property then it is rectangular.

\textbf{Polynomially complete algebras.}
An algebra $(A;F_{A})$ is called \emph{polynomially complete (PC)}
if the clone generated by $F_{A}$ and all constants on $A$ is the clone of all operations on $A$.


\textbf{Linear algebra.}
A finite algebra $(A;w_{A})$ is called \emph{linear} if
it is isomorphic to $(\mathbb{Z}_{p_1}\times\dots\times \mathbb{Z}_{p_s};x_1+\ldots+x_n)$
for prime numbers $p_{1},\ldots,p_{s}$.
It is not hard to show that for every algebra $(B;w_{B})$ there exists a minimal congruence $\sigma$, called
\emph{the minimal linear congruence}, such that
$(B;w_{B})/\sigma$ is linear.

\textbf{Absorption.}
%
Let $\mathbf{B}=(B;F_{B})$ be a subalgebra of $\mathbf{A}=(A;F_{A})$.
We say that $B$ absorbs $A$
if there exists $t\in \Clo(\mathbf{A})$ such that
$t(B,B,\dots,B,A,B,\dots,B) \subseteq B$ for any position of $A$.
In this case we also say that $B$ is an absorbing subuniverse of $\mathbf A$.
If the operation $t$ can be chosen binary or ternary then
$B$ is called \emph{a binary or ternary absorbing subuniverse} of $\mathbf A$.

\textbf{Center.}
Suppose $\mathbf{A} = (A;w_{A})$ is a finite algebra with a special WNU operation.
$C\subseteq A$ is called a \emph{center}
if there exists an algebra $\mathbf{B} = (B;w_{B})$ with a special WNU operation of the same arity and
a subdirect subalgebra $(R;w_{R})$ of $\mathbf{A}\times\mathbf{B}$ such that
there is no binary absorbing subuniverse in $\mathbf{B}$ and
$C = \{a\in A\mid \forall b\in B\colon (a,b)\in R\}.$

\textbf{CSP instance.}
An instance of the constraint satisfaction problem
is called \emph{a CSP instance}.
Sometimes we use the same letter for a CSP instance and for the set of all constraints of this instance.
For a variable $z$ by $D_{z}$ we denote the domain of the variable $z$.

We say that $z_{1}-C_{1}-z_{2}-\dots - C_{l-1}-z_{l}$ is
\emph{a path} in $\Theta$ if $z_{i},z_{i+1}$ are in the scope of $C_{i}$ for every $i$.
We say that \emph{a path $z_{1}-C_{1}-z_{2}-\dots C_{l-1}-z_{l}$  connects $b$ and $c$}
if there exists $a_{i}\in D_{z_{i}}$ for every $i$
such that
$a_{1} = b$, $a_{l} = c$, and
the projection of $C_{i}$ onto $z_{i}, z_{i+1}$
contains the tuple $(a_{i},a_{i+1})$.

A CSP instance is called \emph{1-consistent} if every constraint of the instance is subdirect.
A CSP instance is called \emph{cycle-consistent} if
for every variable $z$ and $a\in D_{z}$
any path starting and ending with $z$ in $\Theta$ 
connects $a$ and $a$.
A CSP instance $\Theta$ is called \emph{linked}
if for every variable $z$ appearing in $\Theta$ and every $a,b\in D_{z}$
there exists a path starting and ending with $z$ in $\Theta$ that connects $a$ and $b$.

Suppose $\mathbf{X'}\subseteq\mathbf{X}$.
Then we can define a projection of $\Theta$ onto $\mathbf{X'}$,
that is a CSP instance where variables are elements of $\mathbf{X'}$ and constraints are projections of the constraints of $\Theta$ onto~$\mathbf{X'}$.
We say that an instance $\Theta$ is \emph{fragmented}
if the set of variables $\mathbf X$ can be divided into 2 nonempty disjoint sets $\mathbf{X_1}$ and
$\mathbf{X_2}$ such that
the constraint scope of any constraint of $\Theta$
either has variables only from $\mathbf{X_1}$, or only from $\mathbf{X_2}$.

A CSP instance $\Theta$ is called \emph{irreducible} if
any instance $\Theta'$ 
such that every constraint of $\Theta'$ is a projection of a constraint from $\Theta$ on some set of variables
is fragmented, linked, or its solution set is subdirect.

\textbf{Weaker constraints.}
We say that a constraint $((y_{1},\ldots,y_{t}),\rho_{1})$ is \emph{weaker than}
a constraint $((z_{1},\ldots,z_s),\rho_{2})$
if $\{y_{1},\ldots,y_{t}\}\subseteq \{z_{1},\ldots,z_s\}$,
$\rho_{2}(z_{1},\ldots,z_s)\rightarrow \rho_{1}(y_{1},\ldots,y_{t})$,
and,
additionally,
$\rho_{1}(y_{1},\ldots,y_{t})\not\rightarrow \rho_{2}(z_{1},\ldots,z_s)$
or $\{y_{1},\ldots,y_{t}\}\neq\{z_{1},\ldots,z_s\}$.
Suppose $((y_{1},\ldots,y_{t}),\rho)$ is a constraint and
$\rho'(y_{1},\ldots,y_{t}) = \exists z \;\rho(z,y_{2},\ldots,y_{t})\wedge \sigma(z,y_{1}),$
where $\sigma$ is a minimal congruence such that $\sigma\supsetneq\ConOne(\rho,1)$.
Then $((y_{1},\ldots,y_{t}),\rho')$ is called \emph{a congruence-weakened constraint}.

\textbf{Minimal linear reduction.}
Suppose the domain set of the instance $\Theta$ is $D = (D_{1},\ldots,D_{n})$.
The domain set $D^{(1)} = (D_{1}^{(1)},\ldots,D_{n}^{(1)})$ is called \emph{a minimal linear reduction} if
$D_{i}^{(1)}$ is an equivalence class of the minimal linear congruence of $D_{i}$ for every $i$.
The reduction $D^{(1)} = (D_{1}^{(1)},\ldots,D_{n}^{(1)})$ is called \emph{1-consistent}
if the instance obtained after the reduction of every domain is 1-consistent.

\textbf{Crucial instances.}
Let $D_{i}^{(1)}\subseteq D_{i}$ for every $i$.
A constraint $C$ of $\Theta$ is called \emph{crucial in $(D_{1}^{(1)},\ldots,D_{n}^{(1)})$}
if $\Theta$ has no solutions in $(D_{1}^{(1)},\ldots,D_{n}^{(1)})$ but
the replacement of $C\in\Theta$ by all
weaker constraints gives an instance with a solution in $(D_{1}^{(1)},\ldots,D_{n}^{(1)})$.
A CSP instance $\Theta$ is called \emph{crucial in $(D_{1}^{(1)},\ldots,D_{n}^{(1)})$} if
every constraint of $\Theta$ is crucial in $(D_{1}^{(1)},\ldots,D_{n}^{(1)})$.
To simplify, instead of ``crucial in $(D_{1},\ldots,D_{n})$'' we say
``crucial''


\section{Algorithm}\label{Algorithm}
In this section we present a modified algorithm from \cite{MyProofCSP}.
The main problem that does not allow to use the original algorithm
for infinite languages is in
Steps~3 and 11, where we replace a constraint by all weaker constraints.
If $\Gamma$ is infinite, we do not have a polynomial upper bound on the number of such replacements.
To fix this problem, instead of replacing a constraint by all weaker constraints
we replace it by all congruence-weakened constraints.
Moreover, to restrict the depth of the recursion we additionally transform our instance to
ensure that all the constraint relations
are essential relations with the parallelogram property.

We start with the new procedures, 
then we explain auxiliary procedures from \cite{MyProofCSP}, 
and finish with the modified main part of the algorithm.

We have made only the following modifications of the algorithm.
\begin{enumerate}
\item We added Step \ref{EssentialRepresentation} to work only with essential relations.
This property is important because, by Lemma~\ref{essentialSize}, any essential relation preserved by an idempotent WNU operation
has exponentially many tuples.

\item We added Step \ref{ParallelogramStep} to ensure that every relation has the parallelogram property.

\item We added Step \ref{IrreducibleCongruenceStep} to ensure that $\ConOne(\rho,1)$ is an irreducible congruence
for every constraint relation $\rho$.

\item We changed Step~\ref{FullSimplificationStep} (Step 3 in \cite{MyProofCSP}). Instead of
replacing every constraint by all weaker constraints we
replace it by all congruence-weakened constraints, and therefore
we increase the congruence $\ConOne(\rho,1)$ for every constraint relation $\rho$.

\item Similarly, we replaced Step 11 in \cite{MyProofCSP} by Steps~\ref{SimplificationCycle} and \ref{SimplificationCycle2}.
In Step \ref{SimplificationCycle}, instead of replacing a constraint by all weaker constraints we replace it by all congruence-weakened constraints.
In Step \ref{SimplificationCycle2}, we try to replace a constraint by all projections onto all variables but one appearing in the constraint.

\item Instead of considering a center we consider a ternary absorption, thus
we replace Steps 4 and 5 in \cite{MyProofCSP} by Step \ref{AbsorptionStep}.

\end{enumerate}


\subsection{New parts}
\textbf{Finding an appropriate projection.}
Suppose $\rho\subseteq A_{1}\times\dots\times A_{n}$,
$\alpha = (a_{1},\ldots,a_{n})\notin\rho$.
Here we explain how
to find a minimal subset $I\subseteq\{1,\ldots,n\}$
such that
$\proj_{I}(\alpha)\notin\proj_{I}(\rho)$.
Note that $\proj_{I}(\rho)$ is always an essential relation.


\begin{enumerate}
\item Put $I: = \varnothing$.
\item Put $k:=1$.
\item While $\proj_{\{1,\ldots,k\}\cup I}(\alpha)\in\proj_{\{1,\ldots,k\}\cup I}(\rho)$ do $k:=k+1$.
\item Put $I:= I\cup \{k\}$
\item If $\proj_{I}(\alpha)\in\proj_{I}(\rho)$,  go to step 2.
\end{enumerate}


\textbf{Essential Representation.}
Suppose $\rho\subseteq A_{1}\times\dots\times A_{n}$.
\emph{An essential representation of $\rho$} is
the following formula
$$\rho(x_{1},\ldots,x_{n}) = \delta_{1}(z_{1,1},\ldots,z_{1,n_{1}})\wedge \dots\wedge\delta_{s}(z_{s,1},\ldots,z_{s,n_{s}}),$$
where $\delta_{i}$ is an essential relation from $\Gamma$,
$z_{i,j}\in\{x_{1},\ldots,x_{n}\}$,
$z_{i,j}\neq z_{i,k}$ for every $i$ and $j\neq k$.
Below we explain how to find a set $G = \{I_{1},\ldots,I_{s}\}$, where $I_{i}\subseteq \{1,\ldots,n\}$ for every $i$,
such that $\proj_{I_{1}}(\rho),\ldots,\proj_{I_{s}}(\rho)$ form an essential representation of $\rho$.
To guarantee this property we require that
every $\proj_{I_{i}}(\rho)$ is essential and
for every $\alpha\in (A_{1}\times \dots\times A_{n})\setminus \rho$
there exists $i$ such that $\alpha\notin\proj_{I_{i}}(\rho)$.

\begin{enumerate}
\item Put $G:= \varnothing$.
\item Choose a tuple $\alpha=(a_{1},\ldots,a_{n})\notin\rho$ such that
$(a_{1},\ldots,a_{n-1},b_{n})\in\rho$ for some $b_{n}$ (we have at most $|\rho|\cdot|A_{n}|$ such tuples).

\item Find a minimal subset $I\subseteq\{1,\ldots,n\}$
such that $\proj_{I}(\alpha)\notin\proj_{I}(\rho)$. Put $G := G\cup\{I\}$.

\item Go to the next tuple in 2).

\item Put $\rho' = \proj_{\{1,\ldots,n-1\}}(\rho)$.
\item By recursive call we calculate the essential representation $G'$ corresponding to $\rho'$.
\item Put $G:=G'\cup G$. 
\item Remove from $G$ all sets that are not maximal in $G$ by inclusion.
\end{enumerate}

Note that the obtained essential representation of an essential relation consists of the original relation.
Therefore, the above procedure can also be used to check whether a relation is essential.

\textbf{Providing the Parallelogram Property.}
In this section we explain how to find the minimal relation $\rho'\supseteq \rho$ having the parallelogram property.
We say that tuples $\alpha_{1},\alpha_{2},\alpha_{3},\alpha_{4}$ \emph{form a rectangle}
if there exists
$I\subseteq \{1,\ldots,n\}$ such that
$\proj_{I}(\alpha_{1}) = \proj_{I}(\alpha_{2})$,
$\proj_{I}(\alpha_{3}) = \proj_{I}(\alpha_{4})$,
$\proj_{\{1,\ldots,n\}\setminus I}(\alpha_{1}) = \proj_{\{1,\ldots,n\}\setminus I}(\alpha_{3})$,
$\proj_{\{1,\ldots,n\}\setminus I}(\alpha_{2}) = \proj_{\{1,\ldots,n\}\setminus I}(\alpha_{4})$.
To find the minimal relation with the parallelogram property
it is sufficient to close the relation under adding the forth tuple of a rectangle.
This can be done in the following way.

For each $\alpha_{1},\alpha_{2},\alpha_{3}\in \rho$.
\begin{enumerate}
\item Put $I_{1}: = \{i\mid \alpha_{1}(i) = \alpha_{2}(i)\}$,
$I_{2}: = \{i\mid \alpha_{1}(i) = \alpha_{3}(i)\}$.
\item If $I_{1}\cup I_{2} = \{1,\ldots, n\}$ then find the forth tuple $\alpha_{4}$ of the corresponding rectangle.
\item If $\alpha_{4}\notin\rho$, add $\alpha_{4}$ to $\rho$.
\end{enumerate}

By Lemma~\ref{essentialSize}, every essential relation has exponentially many tuples, therefore,
this procedure works in polynomial time on the size of $\rho$ if $\rho$ is an essential.

\subsection{Cycle-consistency, non-linked instances, and irreducibility}\label{AlgorithmRemaining}

\textbf{Provide cycle-consistency.}
To provide cycle-consistency it is sufficient to use constraint propagation providing (2,3)-consistency.
Formally, it can be done in the following way.
First, for every pair of variables $(x_{i},x_{j})$ we consider the
intersections of projections of all constraints onto these variables.
The corresponding relation we denote by $\rho_{i,j}$.
For every $i,j,k\in\{1,2,\ldots,n\}$
we replace
$\rho_{i,j}$ by $\rho_{i,j}'$
where $\rho_{i,j}'(x,y) = \exists  z \; \rho_{i,j}(x,y)\wedge \rho_{i,k}(x,z)\wedge \rho_{k,j}(z,y).$
It is not hard to see that this replacement does not change the solution set.

We repeat this procedure while we can change some $\rho_{i,j}$.
If at some moment we get a relation $\rho_{i,j}$ that is not subdirect in $D_{i}\times D_{j}$,
then we can either reduce $D_{i}$ or $D_{j}$, or,
 if $\rho_{i,j}$ is empty, state that there are no solutions.
If we cannot change any relation $\rho_{i,j}$ and every $\rho_{i,j}$ is subdirect in
$D_{i}\times D_{j}$, then the original CSP instance is cycle-consistent.

\textbf{Solve the instance that is not linked.}
Suppose the instance $\Theta$ is not linked and not fragmented, then it can be solved in the following way.
We say that an element $d_{i}\in D_{i}$
and an element $d_{j}\in D_{j}$ are \emph{linked}
if there exists a path that connects $d_{i}$ and $d_{j}$.
Let $P$ be the set of pairs $(i;a)$
such that $i\in\{1,2,\ldots,n\}$, $a\in D_{i}$.
Then $P$ can be divided into the linked components.

It is easy to see that it is sufficient to solve
the problem for
every linked component and join the results.
Precisely, for a linked component
by
$D_{i}'$ we denote the set of all elements $d$
such that
$(i,d)$ is in the component.
It is easy to see that $\varnothing\subsetneq D_{i}'\subsetneq D_{i}$ for every $i$.
Therefore, the reduction to $(D_{1}',\ldots,D_{n}')$ is a CSP instance on smaller domains.

\textbf{Check irreducibility.}
For every $k\in\{1,2,\ldots,n\}$ and every maximal
congruence $\sigma_{k}$ on $D_{k}$ we do the following.
\begin{enumerate}
\item Put $I = \{k\}$.
\item Choose a constraint $C$ having the variable $x_{i}$ in the scope for some $i\in I$, choose another variable $x_{j}$ from
the scope such that $j\notin I$.
\item Denote the projection of $C$ onto $(x_{i},x_{j})$ by $\delta$.
\item Put $\sigma_{j}(x,y) = \exists x'\exists y' \delta(x',x)\wedge \delta(y',y)\wedge \sigma_{i}(x',y')$.
If $\sigma_{j}$ is a proper equivalence relation,
then add $j$ to $I$.
\item go to the next $C$, $x_{i}$, and $x_{j}$ in 2.
\end{enumerate}
As a result we get a set $I$ and
a congruence $\sigma_{i}$ on $D_{i}$ for every $i\in I$.
Put $\mathbf{X'} = \{x_{i}\mid i\in I\}$.
It follows from the construction
that for every equivalence class $E_{k}$ of $\sigma_{k}$
and every $i\in I$
there exists a unique equivalence class $E_{i}$ of $\sigma_{i}$
such that there can be a solution with $x_{k}\in E_{k}$
and $x_{i}\in E_{i}$.
Thus, for every equivalence class of $\sigma_{k}$
we have a reduction to the instance on smaller domains.
Then for every $i$ and $a\in E_{i}$ we consider the corresponding reduction
and check whether there exists a solution with $x_{i} = a$.

Thus, we can check whether the solution set of the projection of the instance onto $\mathbf{X'}$ is subdirect or empty.
If it is empty then we state that there are no solutions.
If it is not subdirect, then we can reduce the corresponding domain.
If it is subdirect, then we
go to the next $k\in\{1,2,\ldots,n\}$ and next maximal
congruence $\sigma_{k}$ on $D_{k}$, and repeat the procedure.

\subsection{Main part}\label{AlgorithmMainPart}

In this section we provide an algorithm that solves $\CSP(\Gamma)$ in polynomial time for constraint languages $\Gamma$ (finite or infinite) that are preserved by an idempotent WNU operation.
We know that $\Gamma$ is also preserved by a special WNU operation $w$.
We extend $\Gamma$ to the set of all relations preserved by $w$.
Let the arity of the WNU $w$ be equal to $m$.
Suppose we have a CSP instance $\Theta = \langle \mathbf{X} , \mathbf{D} , \mathbf{C} \rangle$,
where
$\mathbf{X}=\{x_1,\ldots ,x_n\}$ is a set of variables,
$\mathbf{D}=\{D_{1},\ldots ,D_{n}\}$ is a set of the respective domains,
$\mathbf{C}=\{C_{1},\ldots ,C_{q}\}$ is a set of constraints.

The algorithm is recursive, the list of all possible recursive calls is given in the end of this subsection.
One of the recursive calls is the reduction of a subuniverse $D_{i}$ to $D_{i}'$
such that either $\Theta$ has a solution with $x_{i}\in D_{i}'$,
or it has no solutions at all.


\begin{step}\label{CycleConsistencyStep}
Check whether $\Theta$ is cycle-consistent. If not then we reduce a domain $D_{i}$ for some $i$
or state that there are no solutions.
\end{step}


\begin{step}\label{IrreducibilityStep}
Check whether $\Theta$ is irreducible.  If not then we reduce a domain $D_{i}$ for some $i$
or state that there are no solutions.
\end{step}

\begin{step}\label{EssentialRepresentation}
Replace every constraint by its essential representation.
\end{step}

By Theorem~\ref{ParPropertyForAll}, if $\Theta$ has no solutions then we cannot get a solution while doing the following step.

\begin{step}\label{ParallelogramStep}
Replace every constraint relation by the corresponding constraint relation having the parallelogram property.
If one of the obtained constraint relation is not essential, go to Step~\ref{EssentialRepresentation}.
\end{step}

By Lemma \ref{notIrreducibleImplies}, if $\Theta$ has no solutions then we cannot get a solution in the following step.

\begin{step}\label{IrreducibleCongruenceStep}
If the congruence $\ConOne(\rho,1)$ is not irreducible for some constraint relation $\rho$, then replace the constraint
by the corresponding congruence-weakened constraints and go to Step \ref{EssentialRepresentation}.
\end{step}

At the moment all constraint relations $\rho$ have the parallelogram property and
$\ConOne(\rho,1)$ is an irreducible congruence,
therefore for every constraint there exists a unique congruence-weakened constraint.

\begin{step}\label{FullSimplificationStep}
Replace every constraint of $\Theta$ by the corresponding congruence-weakened constraint,
then replace every constraint relation by the corresponding constraint relation having the parallelogram property.
Recursively calling the algorithm, check that the obtained instance has a solution
with $x_{i}=b$ for every $i\in\{1,2,\ldots,n\}$ and $b\in D_{i}$.
If not, reduce $D_{i}$ to the projection onto $x_{i}$ of the solution set of the obtained instance.
\end{step}


By Theorems \ref{AbsorptionStepThm} and \ref{ternaryAbsorptionStep} we cannot loose the only solution while doing the following step.
\begin{step}\label{AbsorptionStep}
If $D_{i}$ has a binary or ternary absorbing subuniverse $B_{i}\subsetneq D_{i}$ for some $i$, then we reduce $D_{i}$ to $B_{i}$.
\end{step}


By Theorem~\ref{PCStepThm} we can do the following step.
\begin{step}\label{PCStep}
If there exists a congruence $\sigma$ on $D_{i}$ such that
the algebra $(D_{i};w)/\sigma$ is polynomially complete, then we reduce $D_{i}$ to
any equivalence class of $\sigma$.
\end{step}

By Theorem~\ref{NextReduction},
it remains to consider the case when
for every domain $D_{i}$
there exists a congruence $\sigma_{i}$ on $D_{i}$
such that
$(D_{i};w)/\sigma_{i}$ is linear, i.e. it is isomorphic to
$(\mathbb Z_{p_{1}}\times \dots \times \mathbb Z_{p_{l}};x_{1}+\dots+x_{m})$
for prime numbers $p_{1},\ldots,p_{l}$.
Moreover, $\sigma_{i}$ is proper if $|D_{i}|>1$.

We denote $D_{i}/\sigma_{i}$ by $L_{i}$.
We define a new CSP instance $\Theta_{L}$ with domains $L_{1},\ldots,L_{n}$.
To every constraint $((x_{i_1},\ldots,x_{i_s}),\rho)\in \Theta$
we assign a constraint
$((x_{i_1}',\ldots,x_{i_s}'),\rho')$,
where $\rho'\subseteq L_{i_{1}}\times\dots\times L_{i_{s}}$
and $(E_{1},\ldots,E_{s})\in\rho'\Leftrightarrow
(E_{1}\times\dots\times E_{s})\cap\rho\neq\varnothing.$
The constraints of $\Theta_{L}$ are all constraints that are assigned to the constraints of $\Theta$.

Since every relation on $\mathbb Z_{p_{1}}\times \dots \times \mathbb Z_{p_{l}}$ preserved by $x_{1}+\ldots+x_{m}$ is known to be
a conjunction of linear equations,
the instance $\Theta_{L}$ can be viewed as a system of linear equations
in $\mathbb Z_{p}$ for different $p$.

Our general idea is to add some linear equations to $\Theta_{L}$ so that
for any solution of $\Theta_{L}$ there exists the corresponding solution of $\Theta$.
We start with the empty set of equations $Eq$, which is
a set of constraints on $L_{1},\ldots,L_{n}$.

\begin{step}
Put $Eq:=\varnothing$.
\end{step}

\begin{step}\label{AddingEquationCycle}
Solve the system of linear equations $\Theta_{L}\cup Eq$
and choose independent variables
$y_{1},\ldots,y_{k}$.
If it has no solutions then $\Theta$ has no solutions.
If it has just one solution, then, recursively calling the algorithm, solve the reduction of $\Theta$ to this solution.
Either we get a solution of $\Theta$, or $\Theta$ has no solutions.
\end{step}

Then there exist
${Z} = \mathbb Z_{q_{1}}\times \dots \times \mathbb Z_{q_{k}}$ and
a linear mapping $\phi\colon Z \to L_{1}\times\dots\times L_{n}$
such that any solution of $\Theta_{L}\cup Eq$ can be obtained
as $\phi(a_{1},\ldots,a_{k})$ for some $(a_{1},\ldots,a_{k})\in {Z}$.

Note that for any tuple $(a_{1},\ldots,a_{k})\in {Z}$
we can check recursively whether $\Theta$ has a solution
in $\phi(a_{1},\ldots,a_{k})$.
To do this, we just need to solve an easier CSP instance (on smaller domains).
Similarly, we can check whether $\Theta$ has a solution in
$\phi(a_{1},\ldots,a_{k})$
for every $(a_{1},\ldots,a_{k})\in \mathbb {Z}$.
To do this, we just need to check the
existence of a solution in $\phi(0,\ldots,0,1,0,\ldots,0)$ and
$\phi(0,\ldots,0)$  for any position of $1$.

\begin{step}\label{SolveZeros}
Check whether $\Theta$ has a solution in $\phi(0,\ldots,0)$. If it has then stop the algorithm.
\end{step}


\begin{step}
Put $\Theta':= \Theta$. Iteratively remove from $\Theta'$ all
constraints that are weaker than some other constraints of $\Theta'$.
\end{step}

In the following two steps we try to weaken the instance so that it still does not have a solution
in $\phi(a_{1},\ldots,a_{k})$
for some $(a_{1},\ldots,a_{k})\in {Z}$.

\begin{step}\label{SimplificationCycle}
For every constraint $C\in\Theta'$
\begin{enumerate}
\item Let $\Omega$ be obtained from $\Theta'$ by replacing the constraint $C\in\Theta'$ by the corresponding congruence-weakened constraints.
\item Replace every new constraint of $\Omega$ by its essential representation
and remove from $\Omega$ all constraints that are weaker than some other constraints of $\Omega$.
\item If $\Omega$ has no solutions in
$\phi(a_{1},\ldots,a_{k})$
for some $(a_{1},\ldots,a_{k})\in {Z}$,
then put $\Theta':=\Omega$ and repeat Step~\ref{SimplificationCycle}.
\end{enumerate}
\end{step}

\begin{step}\label{SimplificationCycle2}
For every constraint $C\in\Theta'$ 
\begin{enumerate}
\item Let $\Omega$ be obtained from $\Theta'$ by replacing the constraint $C\in\Theta'$ by its projections onto all variables but one appearing in $C$.
\item Replace the new constraints by its essential representation and remove from $\Omega$ all constraints that are weaker than some other constraints of $\Omega$.
\item If $\Omega$ has no solutions in
$\phi(a_{1},\ldots,a_{k})$
for some $(a_{1},\ldots,a_{k})\in {Z}$,
then put $\Theta':=\Omega$ and go to Step~\ref{SimplificationCycle}.
\end{enumerate}
\end{step}


At this moment, the CSP instance $\Theta'$ has the following property.
$\Theta'$ has no solutions in $\phi(b_{1},\ldots,b_{k})$
for some $(b_{1},\ldots,b_{k})\in {Z}$
but if we replace any constraint $C\in\Theta'$ by the corresponding congruence-weakened constraints
then we get an instance that has a solution
in $\phi(a_{1},\ldots,a_{k})$
for every $(a_{1},\ldots,a_{k})\in {Z}$.
Unlike the original algorithm, we cannot claim that $\Theta'$ is crucial in $\phi(b_{1},\ldots,b_{k})$.
Nevertheless, Theorem \ref{addingOneVariable} proves that
we can finish the algorithm in the same way as in the original paper.

In the remaining steps we will find a new linear equation that can be added to $\Theta_{L}$.
Suppose $V$ is an affine subspace of $\mathbb Z_{p}^{h}$ of dimension $h-1$, thus $V$ is the solution set of a linear equation
$c_1x_1 + \dots + c_h x_h = c_{0}$. Then the coefficients $c_0,c_{1},\dots,c_{h}$ can be learned (up to a multiplicative constant) by $(p\cdot h+1)$ queries of the form ``$(a_1,\ldots,a_h) \in V$?'' as follows.
First, we need at most $(h+1)$ queries to find a tuple $(d_{1},\ldots,d_{h})\notin V$. 
Then, to find this equation it is sufficient to check for every $a$ and every $i$
whether the tuple $(d_{1},\ldots,d_{i-1},a,d_{i+1},\ldots,d_{h})$ satisfies this equation.

\begin{step}\label{nonLinkedStep}
Suppose $\Theta'$ is not linked. For each $i$ from $1$ to $k$
\begin{enumerate}
\item Check that for every $(a_{1},\ldots,a_{i})\in \mathbb Z_{q_{1}}\times \dots \times \mathbb Z_{q_{i}}$
there exist $(a_{i+1},\ldots,a_{k})\in \mathbb Z_{q_{i+1}}\times \dots \times \mathbb Z_{q_{k}}$ and  a solution of $\Theta'$
in $\phi(a_{1},\ldots,a_{k})$.
\item If yes, go to the next $i$.
\item If no, then find an equation
$c_{1}y_{1}+\dots+c_{i}y_{i}=c_{0}$ such that
for every $(a_{1},\ldots,a_{i})\in
\mathbb Z_{q_{1}}\times \dots \times \mathbb Z_{q_{i}}$
satisfying $c_{1}a_{1}+\dots+c_{i}a_{i}=c_{0}$
there exist $(a_{i+1},\ldots,a_{k})\in \mathbb Z_{q_{i+1}}\times \dots \times \mathbb Z_{q_{k}}$ and  a solution of $\Theta'$
in $\phi(a_{1},\ldots,a_{k})$.
\item Add the equation $c_{1}y_{1}+\dots+c_{i}y_{i}=c_{0}$ to Eq.
\item Go to Step~\ref{AddingEquationCycle}.
\end{enumerate}
\end{step}


It is not hard to see that $\Theta'$ satisfies the conditions of Theorem~\ref{addingOneVariable}.
Then there exists a constraint $((x_{i_1},\ldots,x_{i_s}),\rho)$ in $\Theta'$ and a relation
$\xi\subseteq D_{i_{1}}\times D_{i_{1}}\times \mathbb Z_{p}$
such that
$(x_{1},x_{2},0)\in \xi\Leftrightarrow (x_{1},x_{2})\in\ConOne(\rho,1)$, $\proj_{1,2}(\xi)\supsetneq\ConOne(\rho,1)$,
and $\ConOne(\rho,1)$ is a congruence.
We add a new variable $z$ with domain $\mathbb Z_{p}$ and a variable $x_{i_{1}}'$ with the same domain as $x_{i_{1}}$.
Then we replace $((x_{i_1},\ldots,x_{i_s}),\rho)$ by
$((x_{i_1}',x_{i_2},\ldots,x_{i_s}),\rho)$ and
add the constraint $((x_{i_{1}},x_{i_{1}}',z),\xi)$.
We denote the obtained instance by $\Upsilon$.
Let $L$ be the set of all tuples $(a_{1},\ldots,a_{k},b)\in
\mathbb Z_{q_{1}}\times \dots \times \mathbb Z_{q_{k}} \times \mathbb Z_{p}$
such that $\Upsilon$ has a solution with $z=b$ in $\phi(a_{1},\ldots,a_{k})$.
By Theorem \ref{ParPropertyMain}, if we replace the constraint relation $\rho$ in $\Theta'$ by
the minimal relation $\rho'\supseteq \rho$ having the parallelogram property
then a minimal linear reduction cannot get a solution after the replacement.
Therefore, $L$ has no tuple $(b_{1},\ldots,b_{k},0)$.
Similarly, if we replace $\rho$ in $\Theta'$ by
$\rho''$ defined by
$$\rho''(x_{i_1},x_{i_2},\ldots,x_{i_s})
= \exists x_{i_1}'\;\rho(x_{i_1}',x_{i_2},\ldots,x_{i_s})\wedge \sigma(x_{i_1},x_{i_1}'),$$
where $\sigma = \proj_{1,2}(\xi)$,
then we get a solution in $\phi(b_{1},\ldots,b_{k})$.
Otherwise, Theorem \ref{ParPropertyMain} implies that
the replacement of $\rho''$ by the minimal relation $\rho'''\supseteq \rho''$ having the parallelogram property
still does not give a solution in $\phi(b_{1},\ldots,b_{k})$.
This contradicts the fact that if we replace any constraint $C\in\Theta'$ by the corresponding congruence-weakened constraints
then we get an instance that has a solution in $\phi(a_{1},\ldots,a_{k})$ for every $(a_{1},\ldots,a_{k})\in {Z}$.
Thus, we know that the projection of $L$ onto the first $k$ coordinates is a full relation.

Therefore, $L$ is defined by one linear equation.
If this equation is $z = b$ for some $b\neq 0$, then
both $\Theta'$ and $\Theta$ have no solutions.
Otherwise, we put $z=0$ in this equation and get
an equation that describes all $(a_{1},\ldots,a_{k})$ such that
$\Theta'$ has a solution in $\phi(a_{1},\ldots,a_{k})$.
It remains to find this equation.


\begin{step}\label{LinkedStep}
Suppose $\Theta'$ is linked.
\begin{enumerate}
\item Find
an equation
$c_{1}y_{1}+\dots+c_{k}y_{k} = c_{0}$
such that
for every $(a_{1},\ldots,a_{k})\in
(\mathbb Z_{q_{1}}\times \dots \times \mathbb Z_{q_{k}})$
satisfying $c_{1}a_{1}+\dots+c_{k}a_{k}=c_{0}$
there exists  a solution of $\Theta'$
in $\phi(a_{1},\ldots,a_{k})$.
\item If the equation was not found then $\Theta$ has no solutions.
\item Add the equation $c_{1}a_{1}+\dots+c_{k}a_{k}=c_{0}$ to Eq.
\item Go to Step~\ref{AddingEquationCycle}.
\end{enumerate}
\end{step}

Note that every time we reduce our domains, we get constraint relations that are still from~$\Gamma$.

We have four types of recursive calls of the algorithm:
\begin{enumerate}
\item we reduce one domain $D_{i}$, for example to a binary absorbing subuniverse
(Steps \ref{CycleConsistencyStep}, \ref{AbsorptionStep}, \ref{PCStep}). 
\item we solve an instance that is not linked. In this case we divide the instance into the linked parts
and solve each of them independently
(Steps \ref{IrreducibilityStep}, \ref{nonLinkedStep}).
\item we replace every constraint by the corresponding congruence-weakened constraint and solve an easier CSP instance
(Step \ref{FullSimplificationStep}).
\item we reduce every domain $D_{i}$ such that
$|D_{i}|>1$
(Steps \ref{AddingEquationCycle}, \ref{SolveZeros}, \ref{SimplificationCycle}, \ref{SimplificationCycle2}, \ref{LinkedStep}).
\end{enumerate}

Lemma~\ref{RecursionTwoDepth} states that 
the depth of the recursive calls of type 3 is at most $2^{|A|^{2}}$.
It is easy to see that the depth of the recursive calls of type 2 and 4 is at most $|A|$.

\section{Correctness of the Algorithm}\label{CorretnessSection}
\subsection{The size of an essential relation and depth of the recursion}

\begin{lem}\label{essentialSize}
Suppose $\rho\subseteq A^{n}$ is an essential relation preserved by a special WNU $w$ of arity $m<n.$
Then $|\rho|\ge 2^{n-m+1}$
\end{lem}

\begin{proof}
By Lemma~\ref{sushnabor}, there exists an essential tuple $(a_{1},a_{2},\ldots,a_{n})$ for $\rho$.
For every $i$ choose $b_{i}$ such that
$$(a_{1},\ldots,a_{i-1},b_{i},a_{i+1},\ldots,a_{n})\in\rho.$$
Put $b_{i}' = w(b_{i},a_{i},\ldots,a_{i})$ for every $i$.
Without loss of generality assume that
$b_{i}' = a_{i}$ for every $i\le k$
and $b_{i}' \neq a_{i}$ for every $i>k$.
Since $w$ preserves $\rho$,
we can show that
$(b_{1}',\ldots,b_{m}',a_{m+1},\ldots,a_{n})\in\rho$.
Hence, $k<m$.
Consider the projection of $\rho$ onto the last $n-m+1$ variables, which we denote by $\rho'$.
It is not hard to check that
$\alpha = (a_{m},\ldots,a_{n})\in\rho'$
and
$(a_{m},\ldots,a_{i-1},b_{i}',a_{i+1}\ldots,a_{n})\in\rho'$
for every $i\in\{m,\ldots,n\}$.

For any subset $I\subseteq \{m,\ldots,n\}$
put
$\alpha_{I} = (c_{m},\ldots,c_{n})\in\rho$,
where $c_{i} = b_{i}'$ if $i\in I$ and
$c_{i} = a_{i}$ otherwise.
We know that $\alpha_{\{i\}}\in\rho'$ for every $i\in\{m,\ldots,n\}$.
Since $w$ is a special WNU,
$w(b_{i}',a_{i},\ldots,a_{i}) = b_{i}'$.
Then we can check that
for any disjoint subsets $I_{1},I_{2}\subseteq \{m,\ldots,n\}$
we have $w(\alpha_{I_{1}},\alpha_{I_2},\alpha,\ldots,\alpha) = \alpha_{I_{1}\cup I_{2}}$.
Thus, $\rho'$ contains the tuple $\alpha_{I}$ for any $I\subseteq \{m,\ldots,n\}$.
Therefore, both $\rho$ and $\rho'$ contain at least $2^{n-m+1}$ tuples.
\end{proof}

\begin{lem}\label{IncreaseAllCon}
Suppose $\rho$ is an essential relation with the parallelogram property,
$\ConOne(\rho,1)$ is an irreducible congruence,
$\sigma = \cover{\ConOne(\rho,1)}$,
$\rho'(y_{1},\ldots,y_{t}) = \exists z \;\rho(z,y_{2},\ldots,y_{t})\wedge \sigma(z,y_{1}).$
Then $\ConOne(\rho',i)\supsetneq \ConOne(\rho,i)$ for every $i$.
\end{lem}

\begin{proof}
Put
$\sigma_{i}(y_{1},y_{1}') =
\exists y_{i}'\exists y_{2}\dots \exists y_{t}\;
\rho(y_{1}',y_{2},\ldots,y_{t})
\wedge
\rho(y_{1},\ldots,y_{i-1},y_{i}',y_{i+1},\ldots,y_{t}).$
Since $\rho$ is an essential relation with the parallelogram property,
we have $\sigma_{i}\supsetneq\ConOne(\rho,1)$.
Since $\ConOne(\rho,1)$ is irreducible, $\sigma_{i}\supseteq\sigma$.
Therefore, for any $(a_{1},a_{1}')\in\sigma\setminus\ConOne(\rho,1)$
there exist $a_{i}',a_{2},a_{3},\ldots,a_{t}$
such that
$(a_{1}',a_{2},\ldots,a_{t}),
(a_{1},\ldots,a_{i-1},a_{i}',a_{i+1},\ldots,a_{t})\in\rho.$
Hence $(a_{i},a_{i}')\in\ConOne(\rho',i)\setminus\ConOne(\rho,i)$ and $\ConOne(\rho',i)\supsetneq \ConOne(\rho,i)$.
\end{proof}

\begin{lem}\label{RecursionTwoDepth}
The depth of the recursive calls of type 3 in the algorithm is less than $2^{|A|^{2}}$.
\end{lem}

\begin{proof}
First, we introduce a partial order on variables of instances.
For a constraint $C$ and a variable $x$ by $MaxComp(C,x)$ we denote the maximal equivalence relation $\sigma$
such that the variable $x$ of $C$ is compatible with $\sigma$.
We say that $x$ in $\Theta$ is \emph{weaker than} $x'$ in $\Theta'$ if
one of the following conditions holds:
\begin{enumerate}
\item the domain of $x$ is a proper subset of the domain of $x'$;
\item the domain of $x$ is equal to the domain of $x'$,
for every $C\in\Theta$ there exists $C'\in\Theta'$ such that
$MaxComp(C,x)\supseteq MaxComp(C',x')$.
\end{enumerate}
Since every constraint relation $\rho$ we have in Step~\ref{FullSimplificationStep}
is an essential relation with the parallelogram property
and $\ConOne(\rho,1)$ is an irreducible congruence,
it follows from Lemma \ref{IncreaseAllCon} that
$\ConOne(C',x)\supsetneq\ConOne(C,x)$
for every new constraint $C'$ we generate from $C$ in Step \ref{FullSimplificationStep}.
Since we provide the parallelogram property after the replacement
in Step \ref{FullSimplificationStep}, we have $MaxComp(C',x)\supsetneq MaxComp(C,x)$.
Thus, every variable is getting weaker.
It is easy to see that any other reduction makes all variables weaker or does not change them.
Therefore, the depth of the recursive calls of type 3 is less than the number of binary relations on the set $A$,
that is $2^{|A|^{2}}$.
\end{proof}

\subsection{Properties of a ternary absorption}

In \cite{MyProofCSP} we proved the following theorem.

\begin{thm}\label{NextReduction}\cite{MyProofCSP}
Suppose $\mathbf{A} = (A;w)$ is an algebra, $w$ is a special WNU of arity $m$.
Then one of the following conditions hold:
\begin{enumerate}
\item there exists a binary absorbing set $B\subsetneq A$,
\item there exists a center $C\subsetneq A$,
\item there exists a proper congruence $\sigma$ on $A$ such that
$(A;w)/\sigma$ is polynomially complete,
\item there exists a proper congruence $\sigma$ on $A$ such that
$(A;w)/\sigma$ is isomorphic to $(\mathbb Z_{p};x_{1}+\dots +x_{m})$.
\end{enumerate}
\end{thm}

Using Corollary 7.9.2 from \cite{MyProofCSP},
this theorem can be rewritten in the following form.
\begin{thm}\label{NextReduction}
Suppose $\mathbf{A} = (A;w)$ is an algebra, $w$ is a special WNU of arity $m$.
Then one of the following conditions hold:
\begin{enumerate}
\item there exists a binary or ternary absorbing set $B\subsetneq A$,
\item there exists a proper congruence $\sigma$ on $A$ such that
$(A;w)/\sigma$ is polynomially complete,
\item there exists a proper congruence $\sigma$ on $A$ such that
$(A;w)/\sigma$ is isomorphic to $(\mathbb Z_{p};x_{1}+\dots +x_{m})$.
\end{enumerate}
\end{thm}

An operation is called \emph{cyclic} if
$f(x_{1},x_{2},\ldots,x_{n}) = f(x_{2},x_{3},\ldots,x_{n},x_{1})$.
We know from \cite{cyclicterms} (see Theorem 4.1) that the existence of a cyclic term is equivalent to the existence of a WNU term.
\begin{thm}\label{cyclicthm}\cite{cyclicterms}
An idempotent algebra $(A;F)$ has a WNU term operation if and only if it has a cyclic term operation of arity $p$ for
every prime number $p>|A|$.
\end{thm}

\begin{lem}\label{terAbsImpliesCenter}
Suppose $B$ absorbs $D$ with a ternary idempotent operation $f$,
$u$ is an idempotent cyclic operation.
Then there exists a cyclic operation $v\in\Clo(\{u,f\})$
preserving the relation $(B\times D)\cup (D\times B)$.
Moreover, $B$ is a center for the algebra $(D;v)$.
\end{lem}

\begin{proof}
By Theorem \ref{cyclicthm} there exists a cyclic operation $u'\in\Clo(u)$ of an odd arity.
Let $n$ be the arity of $u'$.
Let us consider a ternary majority operation $m$ on 2-element set.
We know from \cite{post} that
the $n$-ary majority operation belongs to $\Clo(m)$. 
Consider a term $t$ over $m$ that defines the $n$-ary majority operation.
Replace every operation $m$ in it by $f$ to define an operation $g$ of arity $n$.
It is not hard to see that the operation $g(x_{1},\ldots,x_{n})$
satisfies the following property:
if more than half of the variables are from $B$ then the result is from $B$.
Put
$$v(x_{1},\ldots,x_{n}) =
u'(g(x_{1},\ldots,x_{n}),g(x_{2},\ldots,x_{n},x_{1}),\ldots,g(x_{n},x_{1},\ldots,x_{n-1})).$$
It is not hard to see that $v$ is a cyclic operation preserving
the relation $(B\times D)\cup (D\times B)$.

Let us show that $B$ is a center for the algebra $(D;v)$. Put $E = \{0,1\}$.
Let $v$ be defined on $E$ as a majority operation.
Put $R = (B\times \{0,1\})\cup (D\times \{1\})$.
It is not hard to see that
$v$ preserves $R$, hence $B$ is a center.
\end{proof}

The following three theorems are proved in \cite{MyProofCSP}.

\begin{thm}\label{AbsorptionStepThm}\cite{MyProofCSP}
Suppose $\Theta$ is a cycle-consistent irreducible CSP instance,
$B$ is a binary absorbing set of $D_{i}$.
Then $\Theta$ has a solution if and only if
$\Theta$ has a solution with $x_{i}\in B$.
\end{thm}

\begin{thm}\label{CenterStepThm}\cite{MyProofCSP}
Suppose $\Theta$ is a cycle-consistent irreducible CSP instance,
$C$ is a center of $D_{i}$.
Then $\Theta$ has a solution if and only if
$\Theta$ has a solution with $x_{i}\in C$.
\end{thm}

\begin{thm}\label{PCStepThm}\cite{MyProofCSP}
Suppose $\Theta$ is a cycle-consistent irreducible CSP instance,
there does not exist a binary absorbing subuniverse or a center on $D_{j}$
for every $j$,
$(D_{i};w)/\sigma$ is a polynomially complete algebra, 
$E$ is an equivalence class of $\sigma$.
Then $\Theta$ has a solution if and only if
$\Theta$ has a solution with $x_{i}\in E$.
\end{thm}

Combining Theorem \ref{CenterStepThm} and Lemma \ref{terAbsImpliesCenter} we obtain the following theorem.
\begin{thm}\label{ternaryAbsorptionStep}
Suppose $\Theta$ is a cycle-consistent irreducible CSP instance,
$B$ is a ternary absorbing set of $D_{i}$.
Then $\Theta$ has a solution if and only if
$\Theta$ has a solution with $x_{i}\in B$.
\end{thm}

\subsection{Adding a new linear variable}

To prove the main result of this section we will need the following definitions from \cite{MyProofCSP}.

For an instance $\Omega$ by $\Expanded(\Omega)$ (\emph{Expanded Coverings}) we denote the set of all instances $\Omega'$
such that there exists a mapping $S$ from the set of all variables of $\Omega'$ to the
set of all variables of $\Omega$ satisfying the following conditions:
\begin{enumerate}
\item for every constraint $((x_{1},\ldots,x_{n}),\rho)$ of $\Omega'$
either the variables $S(x_{1}),\ldots,S(x_{n})$ are different and the constraint $((S(x_{1}),\ldots,S(x_{n})),\rho)$ is weaker than
or equal to some constraint of $\Omega$,
or $\rho$ is a binary reflexive relation and $S(x_{1}) = S(x_{2})$;
\item
if a variable $x$ appears in $\Omega$ and $\Omega'$ then $S(x) = x$.
\end{enumerate}


We say that a congruence $\sigma$ is \emph{irreducible} if
it cannot be represented as an intersection of other binary relations $\delta_{1},\ldots,\delta_{s}$ compatible with $\sigma$.
For an irreducible congruence $\sigma$ on a set $A$
by $\cover{\sigma}$ we denote the minimal binary relation $\delta\supsetneq \sigma$ compatible with $\sigma$.

Suppose $\sigma_{1}$ and $\sigma_{2}$ are congruences on $D_{1}$ and $D_{2}$, correspondingly.
A relation $\rho\subseteq D_{1}^{2}\times D_{2}^{2}$ is called \emph{a bridge} from $\sigma_{1}$ to $\sigma_{2}$ if
the first two variables of $\rho$ are compatible with $\sigma_{1}$,
the last two variables of $\rho$ are compatible with $\sigma_{2}$,
$\proj_{1,2}(\rho) \supsetneq \sigma_{1}$,
$\proj_{3,4}(\rho) \supsetneq \sigma_{2}$,
and
$(a_{1},a_2,a_{3},a_{4})\in \rho$ implies
$(a_1,a_2)\in \sigma_{1}\Leftrightarrow (a_3,a_4)\in \sigma_{2}.$
A bridge $\rho\subseteq D^{4}$ is called \emph{reflexive} if
$(a,a,a,a)\in \rho$ for every $a\in D$.

We say that two congruences $\sigma_{1}$ and $\sigma_{2}$ on a set $D$ are \emph{adjacent}
if there exists a reflexive bridge from $\sigma_{1}$ to $\sigma_{2}$.
We say that two constraints $C_{1}$ and $C_{2}$ are \emph{adjacent} in a common variable $x$ if
$\ConOne(C_{1},x)$ and $\ConOne(C_{2},x)$ are adjacent.
An instance is called \emph{connected} if
every constraint in it is  rectangular and
for every two constraints
there exists a path that connects them.

We will need the following statements from \cite{MyProofCSP}.

\begin{lem}\label{CriticalMeansIrreducible}[Lemma 8.3 in \cite{MyProofCSP}]
Suppose $\rho$ is a critical subdirect relation, the $i$-th variable of $\rho$ is rectangular.
Then $\ConOne(\rho,i)$ is an irreducible congruence.
\end{lem}

\begin{thm}\label{ParPropertyForAll}[Theorem 9.5 in \cite{MyProofCSP}]
Suppose $\Theta$ is a cycle-consistent irreducible CSP instance,
its constraint $((x_{1},\ldots,x_{n}),\rho)$ is crucial.
Then $\rho$ is a critical relation with the parallelogram property.
\end{thm}

\begin{thm}\label{ParPropertyMain}[Theorem 9.5 in \cite{MyProofCSP}]
Suppose $D^{(1)}$ is a 1-consistent minimal linear reduction for a cycle-consistent irreducible CSP instance
$\Theta$,
the constraint $((x_{1},\ldots,x_{n}),\rho)$ is crucial in $D^{(1)}$.
Then $\rho$ is a critical relation with the parallelogram property.
\end{thm}

\begin{thm}\label{FindPerfectConstraint}[Theorem 9.8 in \cite{MyProofCSP}]
Suppose $D^{(1)}$ is a 1-consistent minimal linear reduction of a cycle-consistent irreducible CSP instance $\Theta$,
$\Theta$ is crucial in $D^{(1)}$ and not connected.
Then there exists an instance $\Theta'\in\Expanded(\Theta)$ that is crucial in $D^{(1)}$
and contains a linked connected component whose solution set is not subdirect.
\end{thm}

\begin{lem}\label{PathInConnectedComponent}[Corollary 8.15.1 in \cite{MyProofCSP}]
Suppose $\Theta$ is a cycle-consistent linked connected instance
whose constraint relations are critical rectangular relations.
Then for every constraint $C$ and its variable $x$ there exists
a bridge $\delta$ from $\ConOne(C,x)$ to $\ConOne(C,x)$ such that
$\delta(x,x,y,y)$ defines a full relation.
\end{lem}

\begin{lem}\label{LinkedLink}[Corollary 8.10.1 in \cite{MyProofCSP}]
Suppose $\sigma\subseteq A^{2}$ is an irreducible congruence,
$\rho(x_{1},x_{2},y_{1},y_{2})$ is a bridge from
$\sigma$ to $\sigma$ such that $\rho(x,x,y,y)$ defines a full relation.
Then there exists a prime number $p$
and a relation $\zeta\subseteq A\times A\times \mathbb Z_{p}$
such that
$(x_{1},x_{2},0)\in \zeta\Leftrightarrow (x_{1},x_{2})\in\sigma$
and $\proj_{1,2}\zeta = \sigma^{*}$.
\end{lem}

\begin{lem}\label{notIrreducibleImplies}
Suppose $\Theta$ is a cycle-consistent irreducible CSP instance without a solution,
every constraint relation of $\Theta$ has the parallelogram property,
the congruence $\ConOne(\rho,1)$ is not irreducible for some constraint relation $((x_{i_1},\ldots,x_{i_t}),\rho)$,
$\Theta'$ is obtained from $\Theta$ by replacement of $((x_{i_1},\ldots,x_{i_t}),\rho)$ by the corresponding congruence-weakened constraints.
Then $\Theta'$ has no solutions.
\end{lem}

\begin{proof}
Consider binary relations $\delta_{1},\ldots,\delta_{s}$ compatible with $\ConOne(\rho,1)$
such that $\delta_{1}\cap\dots\cap\delta_{s} = \ConOne(\rho,1)$.
Put $\rho_{i}(x_{i_1},\ldots,x_{i_t}) = \exists z \;\rho(z,x_{i_2},\ldots,x_{i_t})\wedge \delta_{i}(z,x_{i_1}),$
and replace the constraint
$((x_{i_1},\ldots,x_{i_t}),\rho)$
by the
constraints
$((x_{i_1},\ldots,x_{i_t}),\rho_{1}),\ldots,((x_{i_1},\ldots,x_{i_t}),\rho_{s})$.
Since $\rho$ has the parallelogram property,
the obtained instance still does not have a solution.
By Theorem~\ref{ParPropertyForAll},
if we replace every new constraint relation by the corresponding relation having the parallelogram property,
then we cannot get a solution.
Therefore, $\Theta'$ has no solutions.
\end{proof}

Similarly to Theorem 9.8 from \cite{MyProofCSP},
we prove the following theorem.

\begin{thm}\label{addingOneVariable}
Suppose the following conditions hold:
\begin{enumerate}
\item $\Theta$ is a linked cycle-consistent irreducible CSP instance with domain set
$(D_{1},\ldots,D_{n})$;
\item there does not exist a binary absorbing subuniverse or a center on $D_{j}$ for every $j$;
\item suppose we replace every constraint of $\Theta$ by the corresponding congruence-weakened constraints,
then replace every constraint relation $\rho$ by the minimal constraint relation $\rho'\supseteq\rho$ having the parallelogram property; then the obtained instance
has a solution with $x_{i} = b$ for every $i$ and $b\in D_{i}$;
\item $D^{(1)}= (D_{1}^{(1)},\ldots,D_{n}^{(1)})$ is a minimal linear reduction of $\Theta$;
\item $\Theta$ has no solutions in $D^{(1)}$;
\item if we replace any constraint by the corresponding congruence-weakened constraints
then the obtained instance has a solution in $D^{(1)}$;
\item if we replace any constraint by its projections onto all variables but one appearing in the constraint
then the obtained instance has a solution in $D^{(1)}$.
\end{enumerate}
Then there exists a constraint $((x_{i_1},\ldots,x_{i_s}),\rho)$ of $\Theta$ and a relation
$\xi\subseteq D_{i_{1}}\times D_{i_{1}}\times \mathbb Z_{p}$
such that
$(x_{1},x_{2},0)\in \xi\Leftrightarrow (x_{1},x_{2})\in\ConOne(\rho,1)$
, $\proj_{1,2}(\xi)\supsetneq\ConOne(\rho,1)$,
and $\ConOne(\rho,1)$ is a congruence.
\end{thm}

\begin{proof}
Assume the contrary.
First, we want to make our instance crucial in $D^{(1)}$.
To do this we replace our constraints by all weaker constraints while we still
do not have a solution in $D^{(1)}$.
The obtained instance we denote by $\Theta'$.
Since $\Theta$ is linked, condition 7 guarantees that $\Theta'$ is also linked.
By Theorem~\ref{ParPropertyMain}, every constraint in $\Theta'$ has the parallelogram property.
Suppose $\Theta'$ is not connected.
Then by Theorem~\ref{FindPerfectConstraint} there exists an instance $\Theta''\in\Expanded(\Theta')$
that is crucial in $D^{(1)}$ and contains a linked connected component $\Omega$ such that
the solution set of $\Omega$ is not subdirect.
By condition 3, there exists a constraint relation $\rho'$ from $\Omega$ such that
for their ancestor $\rho$ from $\Theta$
we have $\ConOne(\rho,1) = \ConOne(\rho',1)$.

If $\Theta'$ is connected, then $\Theta'$ is a linked connected component itself
and we choose a constraint relation $\rho'$ from $\Theta'$
such that $\ConOne(\rho,1) = \ConOne(\rho',1)$, where $\rho$ is the ancestor of $\rho'$ in $\Theta$.
The relations $\rho'$ and $\rho$ can be found because of condition 6.

By Theorem~\ref{ParPropertyMain} and Lemma~\ref{CriticalMeansIrreducible}, $\ConOne(\rho',1)$ is an irreducible congruence.
By Lemma~\ref{PathInConnectedComponent},
there exists a bridge $\delta$
from $\ConOne(\rho',1)$
to $\ConOne(\rho',1)$ such that $\delta(x,x,y,y)$ defines a full relation.
By Lemma~\ref{LinkedLink},
there exists a relation
$\xi\subseteq D_{i_{1}}\times D_{i_{1}}\times \mathbb Z_{p}$
such that
$(x_{1},x_{2},0)\in \xi\Leftrightarrow (x_{1},x_{2})\in\ConOne(\rho',1)$
and $\proj_{1,2}(\xi) = \cover{\ConOne(\rho',1)}$.
\end{proof}


\bibliographystyle{plain}
\bibliography{refs}

\end{document}